\documentclass[11pt]{article}

\usepackage{etex}
\usepackage{bbding}
\usepackage{dingbat}
\usepackage[ruled]{algorithm2e}
\usepackage{amsmath}
\usepackage{algorithmic}
\usepackage{amsfonts}
\usepackage{amsmath}
\usepackage{amssymb}
\usepackage{latexsym}
\usepackage{graphicx}
\usepackage{color}
\usepackage[normalem]{ulem}
\usepackage{balance}
\usepackage{textcomp}
\usepackage{nicefrac}
\usepackage{wasysym}
\usepackage{enumitem}
\usepackage{gensymb}
\usepackage{ctable}
\usepackage{url}
\usepackage{longtable}
\usepackage{citesort}

\usepackage{frcursive}

\usepackage{pifont,epsfig,rotating}

\usepackage{etoolbox}

\DeclareMathAlphabet{\mathpzc}{OT1}{pzc}{m}{it}

\usepackage{etaremune}

\definecolor{dgreyblue}{rgb}{0.26,0.3,0.46}             

\newcommand{\bee}{{\mathbf{e}}}
\newcommand{\cA}{\mathcal{A}}

\newcommand{\cI}{\mathcal{I}}

\newcommand{\cD}{\mathcal{D}}

\renewcommand{\text}[1]{\hbox{\rm \ #1\ \/}}

\newcommand{\be}[1]{\begin{equation}\label{#1}}
\newcommand{\ee}{\end{equation}}
\newcommand{\beqn}{\begin{eqnarray*}}
\newcommand{\eeqn}{\end{eqnarray*}}
\newcommand{\beq}{\begin{eqnarray}}
\newcommand{\eeq}{\end{eqnarray}}
\newcommand{\ben}{\begin{enumerate}}
\newcommand{\een}{\end{enumerate}}
\newcommand{\bi}{\begin{itemize}}
\newcommand{\ei}{\end{itemize}}
\newcommand{\eps}{\varepsilon}
\newcommand{\cU}{\mathcal{U}}

\newcommand{\vecd}{\mathbf{d}}

\newcommand{\IE}{{\em i.e.}\xspace}
\newcommand{\tx}{^{\rm th}}

\newcommand{\refine}{ \prec_{\mathpzc{r}} }

\newtheorem{problem}{Problem}
\newtheorem{theorem}{Theorem}
\newtheorem{remark}{Remark}
\newtheorem{lemma}[theorem]{Lemma}
\newtheorem{corollary}[theorem]{Corollary}

\newenvironment{proof}{{\noindent\bf Proof.\ }}{\hfill{\Pisymbol{pzd}{113}}\vspace{0.1in}}
\newenvironment{proof-sketch}{{\noindent\bf Sketch of Proof.\ }}{\hfill{\Pisymbol{pzd}{113}}\vspace{0.1in}}

\newcommand{\etal}{\emph{et~al.}}

\newcommand{\NP}{\mathsf{NP}}

\newcommand{\nbr}{\mathsf{Nbr}}

\newcommand{\TB}{\vspace{-0.1ex}}\newcommand{\TiE}{\setlength{\itemsep}{-1ex}}

\newcommand{\comment}[1]{}

\newtheorem{proposition}{Proposition}

\newcommand{\EG}{{\it e.g.}\xspace}
\newcommand{\FI}[1]{Fig.~\ref{#1}\xspace}

\newcommand{\iif}{{\bf{if}}}
\newcommand{\tthen}{{\bf{then}}}
\newcommand{\eelse}{{\bf{else}}}
\newcommand{\ffor}{{\bf{for}}}
\newcommand{\wwhile}{{\bf{while}}}
\newcommand{\rreturn}{{\bf{return}}}
\newcommand{\rrepeat}{{\bf{repeat}}}
\newcommand{\ddo}{{\bf{do}}}
\newcommand{\xtc}{{\sc X3c}}
\newcommand{\bSC}{{\sc Sc}}
\newcommand{\sat}{{\sc Sat}}

\newcommand{\diam}{\mathsf{diam}}

\newcommand{\dist}{\mathrm{dist}}

\newcommand{\kopt}{ {k}_{\mathrm{opt}} }
\newcommand{\lopt}{ \mathpzc{L}_{\mathrm{opt}} }
\newcommand{\loptgt}{ \mathpzc{L}_{\mathrm{opt}}^{\geq k} }
\newcommand{\lopteq}{ \mathpzc{L}_{\mathrm{opt}}^{=k} }
\newcommand{\lopteqone}{ \mathpzc{L}_{\mathrm{opt}}^{=1} }

\newcommand{\hatlopteqone}{ \widehat{ \mathpzc{L}_{\mathrm{opt}}^{=1} } }
\newcommand{\hatloptgt}{ \widehat{ \mathpzc{L}_{\mathrm{opt}}^{\geq k} } }
\newcommand{\vopt}{ V_{\mathrm{opt}} }
\newcommand{\voptgt}{ V_{\mathrm{opt}}^{\geq k} }
\newcommand{\vopteq}{ V_{\mathrm{opt}}^{=k} }
\newcommand{\vopteqone}{ V_{\mathrm{opt}}^{=1} }
\newcommand{\vhat}{\widehat{V}}

\newcommand{\hatvopteqone}{ \widehat{ V_{\mathrm{opt}}^{=1} } }
\newcommand{\hatvoptgt}{ \widehat{ V_{\mathrm{opt}}^{\geq k} } }
\newcommand{\hatkopt}{ \widehat{ k_{\mathrm{opt}} } }
\newcommand{\scopt}{ \mathsf{opt}_{\!\!\text{\footnotesize\sc Sc}\!\!}}
\newcommand{\xtcopt}{ \mathsf{opt}_{\!\!\!\text{\footnotesize\sc X3c}\!\!}}

\newcommand{\mds}{{\sc Mds}}

\newcommand{\nokmad}{{\sc Adim}}
\newcommand{\mad}{{\sc Adim}$_{\geq k}$}
\newcommand{\eqmad}{{\sc Adim}$_{= k}$}
\newcommand{\eqmadone}{{\sc Adim}$_{= 1}$}
\newcommand{\eqonemad}{{\sc Adim}$_{=1}$}

\newcommand{\eqdef}{\stackrel{\mathrm{def}}{=}}
\definecolor{columbiablue}{rgb}{0.61, 0.87, 1.0}

\allowdisplaybreaks

\DeclareMathOperator{\adim}{adim}

\title{On the Computational Complexities of Three Privacy Measures for Large Networks Under Active Attack}

\author{
Tanima Chatterjee\thanks{Research partially supported by NSF grant IIS-1160995.}
\hspace*{0.2in}
Bhaskar DasGupta$^\ast$
\hspace*{0.2in}
Nasim Mobasheri$^\ast$
\hspace*{0.2in}
Venkatkumar Srinivasan$^\ast$
\\
Department of Computer Science 
\\
University of Illinois at Chicago 
\\
Chicago, IL 60607, USA
\\
\texttt{\{tchatt2,bdasgup,nmobas2,vsrini7\}@uic.edu}
\and 
Ismael G. Yero\thanks{This research was done while the author was visiting the University of Illinois at Chicago, USA, supported by 
``Ministerio de Educaci\'on, Cultura y Deporte'', Spain, under the 
``Jos\'e Castillejo'' program for young researchers (reference number: CAS15/00007).}
\\
Departamento de Matem\'{a}ticas 
\\
Escuela Polit\'{e}cnica Superior 
\\
Universidad de C\'{a}diz
\\
11202 Algeciras, Spain
\\
\texttt{ismael.gonzalez@uca.es}
}

\begin{document}

\maketitle

\begin{abstract}
With the arrival of modern internet era, large public networks of various types have come 
to existence to benefit the society as a whole and several research areas such as sociology, 
economics and geography in particular. However, the societal and research benefits of these 
networks have also given rise to potentially significant privacy issues in the sense that 
malicious entities may violate the privacy of the users of such a network by analyzing the 
network and deliberately using such privacy violations for deleterious purposes. Such considerations 
have given rise to a new active research area that deals with the quantification of privacy of 
users in large networks and the corresponding investigation of computational complexity issues 
of computing such quantified privacy measures. In this paper, we formalize three such privacy 
measures for large networks and provide non-trivial theoretical computational complexity results 
for computing these measures. Our results show the first two measures can be computed efficiently, 
whereas the third measure is provably hard to compute within a logarithmic approximation factor. 
Furthermore, we also provide computational complexity results for the case when the privacy requirement 
of the network is severely restricted, including an efficient logarithmic approximation.
\end{abstract}

\section{Introduction}

Social networks have certainly become an important center of attention in our modern information society 
by transforming human relationships into a huge interchange of, very often, \emph{sensitive} data.
There are many truly beneficial consequences 
when social network data are released for justified mining and analytical purposes.
For example, researchers in sociology, economics and geography, 
as well as vendors in service-oriented systems and 
internet advertisers can certainly benefit and improve their performances 
by a fair study of the social network data.
But, such benefits are definitely \emph{not} free of cost as 
dishonest individuals or organizations may compromise the \emph{privacy} of its users 
while scrutinizing a public social network and 
may deliberately use such privacy violations 
for harmful or other unfair commercial purposes.
A common way to handle this kind of unwelcome intrusion on the user's privacy 
is to somehow \emph{anonymize} the data by removing most potentially identifying attributes. 
However, even after such anonymization, 
often it may still be possible to infer
many sensitive attributes of a social network that may be linked to its users, 
such as node degrees, inter-node distances or network connectivity, and therefore 
\emph{further} privacy-preserving methods need to be investigated and analyzed.
These additional privacy-preserving methods of social networks are based on the concept of 
$k$-anonymity introduced for microdata in~\cite{Samarati}, 
aiming to ensure that \emph{no} record in a database can be re-identified with a probability higher than $\nicefrac{1}{k}$.

Crucial to modelling a social network anonymization process
are of course the adversary's background knowledge of any object and the structural information about the network that is available. 
For example, assuming the involved social network as a simple graph in which individuals are represented by nodes 
and relationships between pairs of individuals are represented by edges, the adversary's background knowledge 
about a target (a node) could be the node degree~\cite{Liu2008}, the node neighborhood~\cite{Zhou2008}, \emph{etc}. 
In such scenarios, it frequently suffices to develop attacks to re-identify the individuals and their relationships. 
Such attacks are usually called \emph{passive} (see~\cite{5207644} for more information). 
Some examples of passive attacks and the corresponding privacy-preserving methods for social
networks can be found in references~\cite{Liu2008,Zhou2008,Zhou2009}.

In contrast, Backstrom \etal\ introduced the concept of the 
so-called \emph{active} attacks in~\cite{Backstrom}. Such attacks are mainly based on creating and inserting 
in a network some nodes (the ``attacker nodes'') under control by the adversary. 
These attacker nodes could be newly created accounts with pseudonymous or spoofed identities 
(commonly called Sybil nodes), or existing legitimate individuals in the network which are in the adversary's proximity. 
The goal is then to establish links with some other nodes in the network (or even links between other nodes) 
in order to create some sort of ``fingerprints'' in the network that will be further released. 
Clearly, once the releasing action has been achieved, the adversary could retrieve the fingerprints already introduced, 
and use them to re-identify other nodes in the network. 
Backstrom \etal\ in~\cite{Backstrom} showed that 
$O(\sqrt{\log n})$ attacker nodes in a network could in fact seriously compromise the privacy of any arbitrary node.
In recent years, several research works have appeared that deal with decreasing the impact of these 
active attacks (see, for instance, \cite{Viswanath2012}). 
For other related publications on 
privacy-preserving methods in social networks, see~\cite{Netter2011,Wu2010,Zhou2008}.

There are already many well-known active attack strategies for social networks
in order to find all possible vulnerabilities.
However, somewhat surprisingly, 
not many prior research works have addressed the goal of measuring how resistant is a given social network against 
these kinds of active attacks to the privacy. 
To this effect, very recently a novel privacy measure for social networks was introduced in~\cite{Trujillo-1}. 
The privacy measure proposed there was called the 
\emph{$(k,\ell)$-anonymity},
where $k$ is a number indicating a privacy threshold and $\ell$ is the \emph{maximum} number of attacker nodes 
that can be inserted into the network; $\ell$ may be estimated through some statistical methods\footnote{Note that 
other different privacy notions with the \emph{same} name also exists, \EG, 
Feder and Nabar in~\cite{Feder} 
investigated $(k,\ell)$-anonymity where $\ell$ represented the number of common neighbors of two nodes.}.
Trujillo-Rasua and Yero in~\cite{Trujillo-1} showed 
that graphs satisfying $(k,\ell)$-anonymity can prevent adversaries who control at most $\ell$ nodes in the network 
from re-identifying individuals with probability higher than $\nicefrac{1}{k}$. 
This privacy measure relies on a graph parameter called the $k$-metric anti-dimension.

Consider a simple connected unweighted graph $G=(V,E)$ and 
let $\dist_{u, v}$ be the length (number of edges) 
of a shortest path between two nodes $u,v\in V$. 
For an ordered sequence $S =u_1,\dots, u_{t}$ of nodes of $G$ and a node 
$v\in V$, the vector $\vecd_{v,-S}=\left(\dist_{v, u_1}, \dots, \dist_{v, u_{t}}\right)$ is called 
the \emph{metric representation} of $v$ with respect to $S$. Based on the above definition, 
a set $S\subset V$ of nodes is called a $k$-\emph{anti-resolving set} for $G$ 
if $k$ is the largest positive integer such that for every node $v\in V\setminus S$ 
there exist at least $k-1$ different nodes $v_1,\dots,v_{k-1} \in V\setminus S$ such that
$\vecd_{v,-S} = \vecd_{v_1,-S} = \dots = \vecd_{v_{k-1},-S}$, \IE, 
$v$ and $v_1, \dots, v_{k-1}$ have the same metric representation with respect to $S$.
The $k$-\emph{metric anti-dimension} of $G$, denoted by $\adim_k(G)$, is then the minimum cardinality of any 
$k$-anti-resolving set in $G$. Note that $k$-anti-resolving sets may \emph{not} exist in a graph for every $k$.

The connection between $(k, \ell)$-anonymity privacy measure and 
the $k$-metric anti-dimension can be understood in the following way. 
Suppose that an adversary takes control of a set of nodes $S$ of the graph (\IE, $S$ plays the role of attacker nodes), 
and the background knowledge of such an adversary regarding a target node $v$ 
is the metric representation of the node $v$ with respect to $S$. 
The $(k,\ell)$-anonymity privacy measure is then a privacy metric that naturally evolves from the 
adversary's background knowledge. 
Intuitively, if $S$ (the attacker nodes of an adversary) is a $k$-anti-resolving set then 
the adversary cannot uniquely re-identify other nodes in the network (based on the metric representation) 
from these attacker nodes 
with a probability higher than $\nicefrac{1}{k}$ (based on uniform sampling of other nodes), and 
if the $k$-metric anti-dimension of the graph is $\ell$
then the adversary must use at least $\ell$ attacker nodes to 
get the probability of privacy violation down 
to $\nicefrac{1}{k}$.

\subsection{Other Privacy Concepts and Measures}

There is a rich literature on theoretical investigations of privacy measures and 
privacy preserving computational models in several other application areas such as 
multi-party communications, distributed computing and game-theoretic 
settings (\EG, see~\cite{BCEO93,K92,Y79,FJS10,CDSS12}). However, none of these settings apply directly to
our application scenario of active attack model for social networks.
The differential privacy model, introduced by Dwork~\cite{dwork2006} in the context of privacy preservation 
in statistical databases against malicious database queries,
works by computing the correct answer to a query and adding a noise drawn from a specific distribution,
and is quite different from the anonymization
approach studied in this paper.

\subsection{Organization of the Paper}

It is obviously desirable to know how secure a given social network is against active attacks. 
This necessitates the study of computational complexity issues for computing 
$(k,\ell)$-anonymity. 
Currently known results only include some heuristic algorithms with no provable guarantee on performances such as in~\cite{Trujillo-1},
or algorithms for very special cases. In fact, it is not even known if any version of the related computational problems is 
$\NP$-hard. To this effect, we formalize three computational problems related to measuring the 
$(k,\ell)$-anonymity of graphs and 
present non-trivial computational complexity results for these problems.
The rest of the paper is organized as follows:

\begin{enumerate}[label=$\triangleright$]
\item
In Section~\ref{sec-termi}
we review some basic terminologies and notations and then present the three computational problems that we
consider in this paper. For the benefit of the reader, we also briefly review some standard algorithmic complexity
concepts and results that will be used later.
\item
In Section~\ref{sec-statementofresults}, we state the results in this paper mathematically precisely along with some
informal remarks. We group our results based on the problem definitions and the expected size of the attacker nodes.
\item
Sections~\ref{sec-pfb}--\ref{sec-pfe} are devoted to the proofs of the results stated in
Section~\ref{sec-statementofresults}.
\item
We finally conclude 
in Section~\ref{sec-conc}
with some possible future research directions.
\end{enumerate}

\section{Basic Terminologies, Notations and Problem Definitions}
\label{sec-termi}

In this section, we first describe the terminologies and notations required to describe our computational
problems, and subsequently describe several versions of the problems we consider.

\subsection{Basic Terminologies and Notations}
\label{sec-termi-2}

\begin{figure}[ht]
\centering
\includegraphics[width=\linewidth]{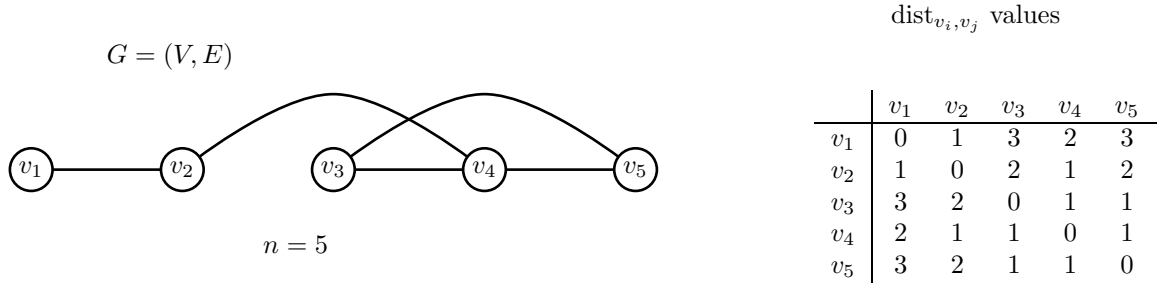}
\caption{\label{fignew}An example to illustrate the notations in Section~\ref{sec-termi-2}.}
\end{figure}

Let $G=(V,E)$ be our \emph{undirected unweighted} input graph over $n$ nodes $v_1,v_2,\dots,v_n$. 
We use $\dist_{v_i,v_j}$ to denote the distance (number of edges in a shortest path) between nodes $v_i$ and $v_j$.
For illustrating various notations, we use the example in \FI{fignew}.
\begin{enumerate}[label=$\blacktriangleright$,leftmargin=*]
%
\item
$\vecd_{v_i}=\left(\dist_{v_i,v_1},\dist_{v_i,v_2},\dots,\dist_{v_i,v_n}\right)$.
For example, $\vecd_{v_2}=(1,0,2,1,2)$.
\item
$\diam(G)=\max\limits_{v_i,v_j\in V} \left\{ \, \dist_{v_i,v_j} \right\}$
is the \emph{diameter} (length of a longest shortest path) of the graph $G=(V,E)$.
For example, $\diam(G)=3$.
\item
$\nbr\left(v_\ell\right)=
\left\{ \, v_j \,|\, \left\{v_\ell,v_j\right\} \in E \, \right\}$
is the (open) \emph{neighborhood} of node $v_\ell$ in $G=(V,E)$.
For example, 
$\nbr\left(v_2\right)= \left\{ v_1,v_4 \right\}
$.
\item
For a subset of nodes $V'\subset V$ and any $v_i\in V\setminus V'$, 
$\vecd_{v_i,-V'}$ 
denotes the metric representation of $v_i$ with respect to $V'$, \IE, 
the vector of $|V'|$ elements obtained from 
$\vecd_{v_i}$ 
by deleting 
$\dist_{v_i,v_j}$ 
for every $v_j\in V\setminus V'$. For example,
$\vecd_{v_2,-\left\{v_1,v_3\right\}}=\left(1,2\right)$.
\item
$\cD_{V'',-V'}=\left\{ \, \vecd_{v_i,-V'} \,|\, v_i\in V'' \, \right\}$
for any $V'' \subseteq V\setminus V'$. 
For example, if  
$V''=\left\{v_2,v_4\right\}$ then 
$\cD_{V'',-\left\{v_1,v_3\right\}}=\big\{(1,2),(2,1) \big\}$.
\item
$\Pi=\left\{ V_1,V_2,\dots,V_k \right\}$ is a partition of $V'\subseteq V$ if and only if 
$\cup_{t=1}^k V_t=V'$ and $V_i\cap V_j=\emptyset$ for $i\neq j$.
\begin{enumerate}[label=$\triangleright$,leftmargin=*]
\item
Partition $\Pi'=\left\{ V_1,V_2,\dots,V_\ell \right\}$ is called a \emph{refinement}\footnote{Our definition is slightly
different from the standard definition of refinement since we have
$
\cup_{t=1}^\ell V_t
\subset
\cup_{t=1}^k V_t$.}
of partition $\Pi$, 
denoted by $\Pi'\refine\Pi$, provided
$
\cup_{t=1}^\ell V_t
\subset
\cup_{t=1}^k V_t
$
and $\Pi'$ can be obtained from $\Pi$ in the following manner:
\begin{enumerate}[label=$\triangleright$,leftmargin=*]
\item
For every node $v_i\in 
\left( \cup_{t=1}^k V_t \right)
\setminus
\left( \cup_{t=1}^\ell V_t \right)
$, 
remove $v_i$ from the set containing it in $\Pi$.
\item
\emph{Optionally}, for every set $V_\ell$ in $\Pi$, replace $V_\ell$ by a partition of $V_\ell$.
\item
Remove empty sets, if any.
\end{enumerate}
For example, 
if 
$\Pi= \Big\{ \left\{ v_1,v_2 \right\}, \left\{ v_3,v_4,v_5 \right\} \Big\}$
and 
$\Pi'= \Big\{ \left\{ v_1,v_2 \right\}, \left\{ v_3 \right\}, \left\{ v_4 \right\} \Big\}$
then 
$\Pi'\refine\Pi$.
\end{enumerate}
\item
The equality relation over 
a set of vectors, all of same length, 
obviously defines 
an \emph{equivalence relation}. The following notations are used for such an equivalence relation
over the set of vectors $\cD_{V\setminus V',-V'}$ for some $\emptyset\subset V'\subset V$. 
\begin{enumerate}[label=$\triangleright$,leftmargin=*]
\item
The set of equivalence classes, which forms a partition of $\cD_{V\setminus V',-V'}$, is denoted by 
$\Pi_{{V\setminus V',-V'}}^{=}$.
For example, 
\\
$\Pi_{\left\{v_1,v_2,v_3 \right\},-\left\{ v_4,v_5\right\}}^{=}
=
\Big\{ \,\big\{ (2,3) \big\}, \, \big\{ (1,2)\big\}, \, \big\{ (1,1)\big\} \, \Big\}$.
\item
Abusing terminologies slightly, two nodes $v_i,v_j\in V\setminus V'$ will be said to
belong to the \emph{same} equivalence class if 
$\vecd_{v_i,-V'}$ and 
$\vecd_{v_j,-V'}$ belong to the same equivalence class in 
$\Pi_{{V\setminus V',-V'}}^{=}$, and thus 
$\Pi_{{V\setminus V',-V'}}^{=}$ 
also defines a partition into equivalence classes of $V\setminus V'$.
For example, 
$\Pi_{\left\{v_1,v_2,v_3 \right\},-\left\{ v_4,v_5\right\}}^{=}$
will also denote
$\Big\{ \, \big\{ v_1 \big\}, \, \big\{ v_2 \big\},\, \big\{ v_3 \big\}\, \Big\}$.
\item
The \emph{measure} of the equivalence relation is defined as 
$
\displaystyle
\mu \left(\cD_{V\setminus V',-V'}\right) 
\eqdef
\min_{\mathcal{Y}\in\Pi_{{V\setminus V',-V'}}^{=}}
\Big\{
\,\left|
\,\mathcal{Y}\,
\,\right|
\Big\}
$.
Thus, 
if a set $S$ is a $k$-anti-resolving set then 
$\cD_{V\setminus S,-S}$ defines a partition into equivalence classes whose measure 
is \emph{exactly} $k$.
For example, 
$\mu\left(\cD_{ \left\{v_1,v_2,v_3 \right\},-\left\{ v_4,v_5\right\} }\right)=1$
and 
$\left\{ v_4,v_5\right\}$ is a 
$1$-anti-resolving set.
\end{enumerate}
\end{enumerate}

\subsection{Problem Definitions}

It is obviously desirable to know how secure a given social network is against active attacks. 
This necessitates the study of computational complexity issues for computing 
$(k,\ell)$-anonymity. 
To this effect, we formalize three computational problems related to measuring the 
$(k,\ell)$-anonymity of graphs
For all the problem versions, let 
$G=(V,E)$ be the (connected undirected unweighted) 
input graph representing the social network under study.

\begin{problem}[metric anti-dimension or \nokmad)]\label{prob1}
Given $G$, find a subset of nodes $V'$ that 
\emph{maximizes} $\mu \left(\cD_{V\setminus V',-V'} \right)$.
\end{problem}

\noindent
{\bf Notation related to Problem~\ref{prob1}}
$\kopt = 
\max\limits_{\emptyset \subset V'\subset V}  
\Big\{
\,
\mu \left(  \cD_{V\setminus V',-V'} \right)
\Big\}
$.

\smallskip

Problem~\ref{prob1} simply finds a $k$-anti-resolving set for the largest possible $k$.
Intuitively, it sets an absolute bound on the privacy violation probability of an adversary assuming 
that the adversary can use \emph{any} number of attacker nodes. In practice, however, the number of attacker
nodes employed by the adversary may be limited, which leads us to the second problem formulation stated below.

\begin{problem}[$k_\geq$-metric anti-dimension or \mad]\label{prob2}
Given $G$ and a positive integer $k$,
find a subset of nodes $V'$ of \emph{minimum} cardinality such that 
$\mu \left(  \cD_{V\setminus V',-V'} \right)\geq k$, if such a $V'$ exists.
\end{problem}

\noindent
{\bf Notation and assumption related to Problem~\ref{prob2}}
$\loptgt=\left| \voptgt \right| = 
\min 
\Big\{
\,|V'|\,\,\,\Big|\,\,\,
\mu \left(  \cD_{V\setminus V',-V'} \right) \geq k
\Big\}
$ for some $\emptyset\subset\voptgt\subset V$.
If $\mu \left(  \cD_{V\setminus V',-V'} \right) \geq k$ for \emph{no} $V'$ 
then we set $\loptgt=\infty$ and $\voptgt=\emptyset$. 

\smallskip

Problem~\ref{prob2} finds a $k$-anti-resolving set for largest $k$
while simultaneously minimizing the number of attacker nodes. 

The remaining third version of our problem formulation relates to a trade-off between 
privacy violation probability and the corresponding minimum number of attacker nodes needed to achieve such a violation.
To understand this motivation, suppose that $G$ has 
a $k$-metric anti-dimension of $\ell$,  
a $k'$-metric anti-dimension of $\ell'$, 
$k'>k$ and $\ell'<\ell$.
Then, this provides a trade-off between privacy and number of attacker nodes, namely we may
allow a smaller privacy 
violation probability $\nicefrac{1}{k'}$ but the network can tolerate \emph{adversarial control}
of a \emph{fewer} number $\ell'$ of nodes or 
we may
allow a larger privacy 
violation probability $\nicefrac{1}{k}$ but the network can tolerate adversarial 
control of a larger number $\ell$ of nodes.
Such a trade-off may be crucial for a network administrator in administering privacy of a network
or for an individual in its decision to join a network.
Clearly, this necessitates solving a problem of the following type.

\begin{problem}[$k_{=}$-metric antidimension or \eqmad]\label{prob3}
Given $G$ and a positive integer $k$,
find a subset of nodes $V'$ of \emph{minimum} cardinality such that 
$\mu \left(  \cD_{V\setminus V',-V'} \right)=k$, if such a $V'$ exists.
\end{problem}

\noindent
{\bf Notation and assumption related to Problem~\ref{prob3}}
$\lopteq=\left| \vopteq \right| = 
\min 
\Big\{
\,|V'|\,\,\,\Big|\,\,\,
\mu \left(  \cD_{V\setminus V',-V'} \right)=k
\Big\}
$ for some $\emptyset\subset\vopteq\subset V$.
If $\mu \left(  \cD_{V\setminus V',-V'} \right)=k$ for \emph{no} $V'$ 
then we set $\lopteq=\infty$ and $\vopteq=\emptyset$

\subsection{Standard Algorithmic Complexity Concepts and Results}

For the benefit of the reader, 
we summarize the following concepts and results from the computational complexity theory domain.
\emph{We assume that the reader is familiar with standard $O$, $\Omega$, $o$ and $\omega$ notations used
in asymptotic analysis of algorithms $($\EG, see~{\rm\cite{CLR90}}$)$}.

An algorithm $\cA$ for a minimization (resp., maximization) problem is said to have an \emph{approximation ratio} of $\eps$ 
(or is simply an \emph{$\eps$-approximation})~\cite{V01} provided $\cA$ runs in polynomial time in the size of its input and produces 
a solution with an objective value 
\emph{no larger than} $\eps$ times 
(resp., \emph{no smaller than} $\nicefrac{1}{\eps}$ times) 
the value of the optimum.
DTIME$\left(n^{\log\log n}\right)$ 
refers to the class of problem that can be solved by a deterministic algorithm
running in 
$(n^{\log\log n})$ time when $n$ is the size of the input instance; it is widely believed that 
$\NP\not\subset$DTIME$(n^{\log\log n})$.

The minimum set-cover problem (\bSC) is a well-known combinatorial problem that
is defined as follows~\cite{CLR90,GJ79}. Our input is an 
universe $\cU=\left\{a_1,a_2,\dots,a_n\right\}$
of $n$ elements, and a collection of $m$ sets 
$S_1,S_2,\ldots,S_m\subseteq \cU$
over this universe 
with $\cup_{j=1}^m S_j=U$.
A valid solution of \bSC\ 
is a subset of indices $\cI\subseteq \{1,2,\dots,m\}$ such
that every element in $\cU$ is ``covered'' by a set whose index is in $\cI$, \IE,
$\forall \, a_j\in \cU \,\, \exists\, i\in \cI  \;:\;  a_j\in S_i$.  
The objective of \bSC\ is to 
{\em minimize} the number $|\cI|$ of selected sets. 
We use the notation $\scopt$ to denote the size (number of sets) in an optimal solution of an instance of \bSC.
On the inapproximability side, \bSC\ is $\NP$-hard~\cite{GJ79} and,  
assuming $NP\not\subseteq\,$DTIME$\left(n^{\log\log n}\right)$, 
\bSC\ does not admit a 
$(1-\eps)\ln n$-approximation for any constant $0<\eps<1$~\cite{F98}.
On the algorithmic side, \bSC\ admits a $(1+\ln n)$-approximation using a simple 
greedy algorithm~\cite{J74}
that can be easily implemented 
to run in $O\big(\sum_{i=1}^m \left| S_i \right| \,\big)$ time~\cite{CLR90}.

\section{Our Results}
\label{sec-statementofresults}

In this section we provide precise statements of our results, leaving their proofs in 
Sections~\ref{sec-pfb}--\ref{sec-pfe}.

\subsection{Polynomial Time Solvability of \nokmad\ and \mad}

\begin{theorem}\label{thm1}~\\
\noindent
{\bf (a)} 
Both \nokmad\ and \mad\ 
can be solved in $O\left(n^4\right)$ time.

\smallskip
\noindent
{\bf (b)} 
Both \nokmad\ and \mad\ 
can also be solved in $O\left(\frac{n^4 \,\log n}{k}\right)$ time ``with high 
probability'' $($\IE, 
with a probability of at least $1-n^{-c}$ for some constant $c>0$$)$.
\end{theorem}

\begin{remark}
The randomized algorithm in Theorem~{\em\ref{thm1}(b)} runs faster that the deterministic algorithm in Theorem~{\em\ref{thm1}(a)} 
provided $k=\omega(\log n)$.
\end{remark}

\subsection{Computational Complexity of \eqmad}

\subsubsection{The Case of Arbitrary $k$}

\begin{theorem}\label{k-large}~\\
\noindent
{\bf (a)} 
\eqmad\ is $\NP$-complete for any integer 
$k$ in the range 
$1\leq k \leq n^\eps$ where 
$0\leq \eps<\frac{1}{2}$ is any arbitrary constant, 
even if the diameter of the input graph is $2$.

\medskip
\noindent
{\bf (b)} 
Assuming $\NP\not\subseteq\,\,$\emph{DTIME}$\,(n^{\log\log n})$, 
there exists a universal constant $\delta>0$ such that 
\eqmad\ does not admit a $\left(\frac{1}{\delta}\ln n\right)$-approximation 
for any integer 
$k$ in the range 
$1\leq k \leq n^\eps$ where 
$0\leq \eps<\frac{1}{2}$ is any arbitrary constant, 
even if the diameter of the input graph is $2$.

\medskip
\noindent
{\bf (c)} 
If $k=n-c$ for some constant $c$
then $\lopteq=c$ if a solution exists and 
\eqmad\ can be solved in polynomial time. 
\end{theorem}

\begin{remark}\label{rem-opt}~\\
\noindent
{\bf (a)} 
For $k=1$, the inapproximability ratio 
in Theorem~\emph{\ref{k-large}(a)}
is asymptotically optimal up to a constant factor because of the 
$(1+\ln (n-1))$-approximation of \eqonemad\ in Theorem~\emph{\ref{k-one}(a)}.

\medskip
\noindent
{\bf (b)} 
The result
in Theorem~\emph{\ref{k-large}(b)}
provides a much stronger inapproximability result compared to that 
in Theorem~\emph{\ref{k-large}(a)}
at the expense of a slightly weaker complexity-theoretic assumption 
$($\IE, $\,\NP\not\subseteq\,\,$\emph{DTIME}$\,(n^{\log\log n})$ vs. \mbox{P}$\,\neq\NP$$)$.
\end{remark}

\subsubsection{The Case of $k=1$}

Note that even when $k=1$ \eqmad\ is $\NP$-hard and even hard to approximate within a logarithmic factor due
to Theorem~\ref{k-large}.
We show the following algorithmic results for \eqmad\ when $k=1$.

\begin{theorem}\label{k-one}~\\
\noindent
{\bf (a)} 
\eqmadone\ admits a $(1+\ln (n-1)\,)$-approximation in $O\left(n^3\right)$ time.

\medskip
\noindent
{\bf (b)} 
If $G$ has at least one node of degree $1$ then 
$\lopteqone=1$ and thus \eqmadone\ can be solved in $O\left(n^3\right)$ time.

\medskip
\noindent
{\bf (c)} 
If 
$G$ does not contain a cycle of $4$ edges then 
$\lopteqone\leq 2$ and thus \eqmadone\ can be solved in $O\left(n^3\right)$ time.
\end{theorem}


\section{Proof of Theorem~\ref{thm1}}
\label{sec-pfb}

\noindent
{\bf (a)}
We first consider the claim for \mad.
We begin by proving some structural properties of valid solutions for \mad.

\begin{proposition}\label{prop1}
Consider two subsets of nodes $\emptyset\subset V_1\subset V_2\subset V$.
Let $v_i,v_j\in V_2$ be two nodes such that they do not belong to the same equivalence class in  
$\Pi_{{V\setminus V_1,-V_1}}^{=}$.
Then $v_i$ and $v_j$ do not belong to the same equivalence class in 
$\Pi_{{V\setminus V_2,-V_2}}^{=}$ also.
\end{proposition}

\begin{proof}
Since $v_i$ and $v_j$ are not in the same equivalence class in  
$\Pi_{{V\setminus V_1,-V_1}}^{=}$, 
we have 
$\vecd_{v_i,-V_1}\neq\vecd_{v_j,-V_1}$ which in turn implies (since $V_1\subset V_2$) 
$\vecd_{v_i,-V_2}\neq\vecd_{v_j,-V_2}$ which implies 
$v_i$ and $v_j$ are not in the same equivalence class in  
$\Pi_{{V\setminus V_2,-V_2}}^{=}$.
\end{proof}

\begin{corollary}\label{cor1}
Proposition~\ref{prop1} implies 
$\Pi_{{V\setminus V_2,-V_2}}^{=}\refine\Pi_{{V\setminus V_1,-V_1}}^{=}$.
\end{corollary}

Note that 
$\Pi_{{V\setminus V_2,-V_2}}^{=}\refine\Pi_{{V\setminus V_1,-V_1}}^{=}$
in Corollary~\ref{cor1} does not necessarily imply that 
$\mu \left( \cD_{V\setminus V_2,-V_2} \right)\leq
\mu \left( \cD_{V\setminus V_1,-V_1} \right)$.
The following proposition gives some condition for this to happen.

\begin{proposition}\label{prop2}
Consider two subsets of nodes $\emptyset\subset V_1\subset V_2\subset V$, and 
let $S_1,S_2,\dots,S_\ell\subseteq V\setminus V_1$ be the only $\ell>0$ equivalence classes $($subsets of nodes$)$ 
in $\Pi_{{V\setminus V_1,-V_1}}^{=}$
such that 
$\left|S_1\right|=\left|S_2\right|=\dots=\left|S_\ell\right|=\mu \left( \cD_{V\setminus V_1,-V_1} \right)$.
Then, 
\begin{quote}
\begin{enumerate}[label=$\triangleright$]
\item
$\cup_{t=1}^\ell S_t \not\subseteq V_2\setminus V_1$ implies 
$\mu \left( \cD_{V\setminus V_2,-V_2} \right) \leq
\mu \left( \cD_{V\setminus V_1,-V_1} \right)$, and 
\item
if 
$\emptyset\subset V_2\cap S_j\subset S_j$ 
for some $j\in\{1,\dots,\ell\}$
then 
$\mu \left( \cD_{V\setminus V_2,-V_2} \right)<
\mu \left( \cD_{V\setminus V_1,-V_1} \right)$.
\end{enumerate}
\end{quote}
\end{proposition}

\begin{proof}
Since  
$V_2\cap S_j\subset S_j$, 
there exists a node $v_p$ such that 
$v_p\in S_j$ and 
$v_p\notin V_2$.
Similarly, 
since  
$\emptyset\subset V_2\cap S_j$, 
there exists a node $v_q$ such that 
$v_q\in S_j$ and 
$v_q\in V_2$.
By Corollary~\ref{cor1}, 
$\Pi_{{V\setminus V_2,-V_2}}^{=}\refine\Pi_{{V\setminus V_1,-V_1}}^{=}$ and 
thus the following implications hold:
\begin{itemize}
\item
If $\cup_{t=1}^\ell V_t \not\subseteq V_2\setminus V_1$ then 
$\Pi_{{V\setminus V_2,-V_2}}^{=}$ 
contains an equivalence class (subset of nodes) $S_{j'}\subseteq S_j$ such that $v_i\in S_{j'}$. 
This implies
$\mu \left( \cD_{V\setminus V_2,-V_2} \right) \leq
\left| S_{j'} \right| \leq \left| S_j \right| = 
\mu \left( \cD_{V\setminus V_1,-V_1} \right)$.
\item
If
there exists a $S_j$ such that 
$\emptyset\subset V_2\cap S_j\subset S_j$ then 
$\Pi_{{V\setminus V_2,-V_2}}^{=}$ 
contains an equivalence class $\emptyset\subset S_{j'}\subset S_j$ with $v_p\in S_{t'}$. 
This implies
$\mu \left( \cD_{V\setminus V_2,-V_2} \right) \leq
\left| S_{j'} \right| < \left| S_j \right| = 
\mu \left( \cD_{V\setminus V_1,-V_1} \right)$.
\end{itemize}
\end{proof}

Based on the above structural properties, we design Algorithm~I for \mad\ as shown below.

\medskip
\begin{longtable}{r l l }
%
\toprule
\multicolumn{3}{c}{Algorithm I: $O\left(n^4\right)$ time deterministic algorithm for \mad.}
\\
\midrule
{\bf 1.} & 
\multicolumn{2}{l}{
Compute $\vecd_i$ for all $i=1,2,\dots,n$ in $O\left(n^3\right)$ time using Floyd-Warshall algorithm~\cite[p. 629]{CLR90}
}
\\
[5pt]
{\bf 2.} & 
\multicolumn{2}{l}{
$\hatloptgt\leftarrow\infty$ ; 
$\hatvoptgt\leftarrow\emptyset$
}
\\
[5pt]
{\bf 3.} & 
\multicolumn{2}{l}{
\ffor\ each $v_i\in V$ \ddo 
  \hspace*{0.2in} $(*$ we \emph{guess} $v_i$ to belong to $\voptgt$ $*)$
}
\\
[5pt]
{\bf 3.1} & & 
  $V'=\left\{ v_i \right\}$ ; $\mathsf{done}\leftarrow\mathsf{FALSE}$
\\
[5pt]
{\bf 3.2} & & 
\wwhile\ $\big($ $(V\setminus V'\neq\emptyset)$ \textsf{AND} $(\mathsf{NOT}\,\,\mathsf{done})$ $\big)$ \ddo 
\\
[5pt]
{\bf 3.2.1}
& & 
\hspace*{0.2in}
compute 
$\mu \left(  \cD_{V\setminus V',-V'} \right)$
\\
[5pt]
{\bf 3.2.2}
& & 
\hspace*{0.2in}
\iif\ 
$\Big( \, \big( \, \mu \left(  \cD_{V\setminus V',-V'} \right)\geq k \, \big)$ and 
$\big( \, |V'|<\hatloptgt \, \big) \, \Big)$ 
\\
{\bf 3.2.3}
& & 
\hspace*{0.4in}
\tthen\
$\,\,\,\,\,$
$\hatloptgt\leftarrow|V'|$ ; 
$\hatvoptgt\leftarrow V'$ ; 
$\mathsf{done}\leftarrow\mathsf{TRUE}$
\\
[5pt]
{\bf 3.2.4}
& & 
\hspace*{0.4in}
\eelse\
$\,\,\,\,\,\,\,\,$
let $V_1,V_2,\dots,V_\ell$ be the \emph{only} $\ell>0$ equivalence classes (subsets of nodes) 
\\
& & 
\hspace*{1.4in}
in $\Pi_{{V\setminus V',-V'}}^{=}$
such that 
$\left|V_1\right|=\left|V_2\right|=\dots=\left|V_\ell\right|=\mu \left( \cD_{V\setminus V',-V'} \right)$
\\
[5pt]
{\bf 3.2.5}
& & 
\hspace*{0.85in}
$V'\,\leftarrow\,V'\cup \left( \cup_{t=1}^\ell V_t \right)$
\\
[5pt]
{\bf 4.} & 
\multicolumn{2}{l}{
\rreturn\ $\hatloptgt$ and $\hatvoptgt$ as our solution
}
\\
\bottomrule
\end{longtable}
\medskip

\begin{lemma}[Proof of correctness]
Algorithm I returns an optimal solution for \mad.
\end{lemma}

\begin{proof}
Assume that $\voptgt\neq\emptyset$ since otherwise obviously our returned solution is correct.
Fix any optimal solution (subset of nodes) $\voptgt$ of measure 
$\mu \left(  \cD_{V\setminus \voptgt,-\voptgt} \right)\geq k$ 
and 
select any arbitrary node $v_\ell\in\voptgt$.
Consider the iteration of the \ffor\ loop in Step~3 when $v_i$ is equal to $v_\ell$.
We now analyze the run of \emph{this particular iteration}.

Let $\left\{v_\ell\right\}=V_1\subset V_2\subset\dots\subset V_\kappa$
be the $\kappa$ subsets of nodes that were assigned to $V'$ in \emph{successive} iterations of 
the \wwhile\ loop in Step~3.2.
We have the following cases to consider.
\begin{description}
\item[Case 1: $\voptgt=V_t$ for some $t\in\{1,2,\dots,\kappa\}$.]
Then, 
our solution is a set $\hatvoptgt$ such that 
\\
$\mu \left(  \cD_{V\setminus \hatvoptgt,-\hatvoptgt} \right)\geq k$ and 
$\hatloptgt\leq\loptgt$.
\item[Case 2: $\voptgt\neq V_t$ for any $t\in\{1,2,\dots,\kappa\}$.]
Since 
$V_1=\left\{v_\ell\right\}\subset\voptgt$ and 
$V_t\neq\voptgt$ for any $t\in\{1,2,\dots,\kappa\}$, 
only one of the following cases is possible:
\begin{description}
\item[Case 2.1: 
there exists $r\in\{1,2,\dots,\kappa-1\}$ such that 
$V_r\subset \voptgt$ but $V_{r+1}\not\subseteq\voptgt$.
]
Let 
\\
$V_{r,1}, V_{r,2},\dots,V_{r,p}\subseteq V\setminus V_r$ be all the $p>0$ equivalence classes (subsets of nodes) 
in $\Pi_{{V\setminus V_r,-V_r}}^{=}$
such that 
$\left|V_{r,1}\right|=\left|V_{r,2}\right|=\dots=\left|V_{r,p}\right|=\mu \left( \cD_{V\setminus V_r,-V_r} \right)$.
Now we note the following:
\begin{itemize}
\item
By Step~3.2.5, 
$V_{r+1}=V_r\cup V_{r,1}\cup V_{r,2}\cup\dots\cup V_{r,p}$.
\item
Thus,
$V_r\subset \voptgt$ 
and 
$V_{r+1}\not\subseteq\voptgt$
implies 
$
V_{r,1}\cup V_{r,2}\cup\dots\cup V_{r,p}
\not\subseteq\voptgt$, 
and therefore 
there exists an index $1\leq s\leq p$ such that 
$
Z=
V_{r,s}\setminus
\voptgt \neq\emptyset
$.
Let 
$
Z'=
V_{r,s}
\setminus
Z
$ ($Z'$ could be empty).
Then, 
for some $\emptyset\subset Z''\subseteq Z$, 
$Z''$ is an equivalence class in  
$
\Pi_{{V\setminus \left( V_r \cup Z' \right), -\left(V_r \cup Z' \right)}}^{=}$
implying 
\begin{gather}
\mu \left(  \cD_{V\setminus \left( V_r \cup Z' \right),-\left( V_r \cup Z' \right) } \right)
\leq 
\left| Z'' \right| 
\leq 
|Z|
\label{eq2}
\end{gather}
Since 
$V_r\cup Z'\subseteq\voptgt$, 
we have
\begin{multline*}
\underset{
\text{\footnotesize (in Corollary~\ref{cor1}, set $V_2=\voptgt$ and $V_1=V_r\cup Z'$)}
}
{
\Pi_{{V\setminus \voptgt,-\voptgt}}^{=}
\refine
\Pi_{{V\setminus \left( V_r \cup Z' \right), -\left(V_r \cup Z' \right)}}^{=}
}
\\
\Rightarrow\,
k
\!
\leq
\!
\mu \left(  \cD_{V\setminus \voptgt,-\voptgt} \right)
\leq 
\mu \left(  \cD_{V\setminus \left( V_r \cup Z'\right),-\left( V_r \cup Z'\right)} \right)
\!\!
\!\!
\!\!
\!\!
\!\!
\underset{
\text{\begin{tabular}{c} \footnotesize by  \footnotesize\eqref{eq2} \end{tabular}}
}
{
\leq 
}
\!\!
\!\!
\!\!
\!\!
\!\!
|Z|
\leq 
\left| V_{r,s} \right| 
=
\mu \left( \cD_{V\setminus V_r,-V_r} \right)
\end{multline*}
\end{itemize}
Thus,
$\mu \left( \cD_{V\setminus V_r,-V_r} \right)\geq k$ and 
$\left| V_r \right|<\left| \voptgt\right| = \loptgt$, contradicting the optimality of $\loptgt$.
\item[Case 2.2: $V_{\kappa}\subset\voptgt$.]
If $\mathsf{done}$ was set to $\mathsf{TRUE}$ at the last iteration of the \wwhile\ loop, then 
$\mu \left( \cD_{V\setminus V_{\kappa},-V_{\kappa}} \right)\geq k$ and 
$\left| V_{\kappa} \right|<\left| \voptgt\right| = \loptgt$, contradicting the optimality of $\loptgt$.
Thus, 
$\mathsf{done}$ 
must have remained $\mathsf{FALSE}$ after the last iteration of the \wwhile\ loop, which 
implies 
$\mu \left( \cD_{V\setminus V_\kappa,-V_\kappa} \right)<k$.
Let
$V_{\kappa,1},V_{\kappa,2},\dots,V_{\kappa,p}\subseteq V\setminus V_\kappa$ be all the $p>0$ equivalence classes (subsets of nodes) 
in $\Pi_{{V\setminus V_\kappa,-V_\kappa}}^{=}$
such that 
$\left|V_{\kappa,1}\right|=\left|V_{\kappa,2}\right|=\dots=\left|V_{\kappa,p}\right|=
\mu \left( \cD_{V\setminus V_\kappa,-V_\kappa} \right)$.
Since 
$V_{\kappa}\subset\vopt$, we have 
\begin{multline*}
\underset{
\text{\footnotesize (in Corollary~\ref{cor1}, set $V_2=\vopt$ and $V_1=V_\kappa$)}
}
{
\Pi_{{V\setminus \vopt,-\vopt}}^{=}
\refine
\Pi_{{V\setminus V_\kappa, -V_\kappa }}^{=}
}
\\
\!\!\!\!\!
\Rightarrow\,\,\,
k
\!
\leq
\!
\mu \left(  \cD_{V\setminus \vopt,-\vopt} \right)
\leq 
\mu \left(  \cD_{V\setminus V_\kappa ,- V_\kappa } \right)
\!\!
\!\!
\!\!
\!\!
\!\!
\underset{
\text{\begin{tabular}{c} \footnotesize by  \footnotesize\eqref{eq2} \end{tabular}}
}
{
\leq 
}
\!\!
\!\!
\!\!
\!\!
\!\!
|Z|
\leq 
\left| V_{\kappa,p} \right| 
=
\mu \left( \cD_{V\setminus V_\kappa,-V_\kappa} \right)
\end{multline*}
Thus,
$\mu \left( \cD_{V\setminus V_\kappa,-V_\kappa} \right)\geq k$ 
contradicting our assumption of 
$\mu \left( \cD_{V\setminus V_\kappa,-V_\kappa} \right)<k$. 
\end{description}
\end{description}
\end{proof}

\begin{lemma}[Proof of time complexity]
Algorithm I runs in 
$O\left(n^4\right)$ time.
\end{lemma}

\begin{proof}
There are $n$ choices for the \ffor\ loop in Step~3. For each such choice, we analyze the execution of the 
\wwhile\ loop in Step~3.2.
The running time in each iteration of the \wwhile\ loop is dominated by the time taken to compute 
$\Pi^{=}_{V\setminus \left( V'\cup\, \left( \cup_{t=1}^{\ell}V_t\right)\,\right) ,-V'\cup\,\left( \cup_{t=1}^{\ell}V_t\right) }$
from 
$\Pi^{=}_{V\setminus V',-V'}$.
Suppose that 
$\cup_{t=1}^{\ell}V_t=\left\{v_{i_1},v_{i_2},\dots,v_{i_p}\right\}$.
By Corollary~\ref{cor1}, 
\begin{multline*}
\Pi^{=}_{V\setminus \left( V' \cup \left\{v_{i_1},v_{i_2},\dots,v_{i_p-1},v_{i_p}\right\} \right),-V'\cup \left\{v_{i_1},v_{i_2},\dots,v_{i_p-1},v_{i_p}\right\} }
\refine
\Pi^{=}_{V\setminus \left( V' \cup \left\{v_{i_1},v_{i_2},\dots,v_{i_p-1}\right\} \right),-V'\cup \left\{v_{i_1},v_{i_2},\dots,v_{i_p-1}\right\} }
\\
\refine
\dots
\refine
\Pi^{=}_{V\setminus \left( V' \cup \left\{v_{i_1},v_{i_2}\right\} \right),-V'\cup \left\{v_{i_1},v_{i_2}\right\} }
\refine
\Pi^{=}_{V\setminus \left( V' \cup \left\{v_{i_1}\right\} \right),-V'\cup \left\{v_{i_1}\right\} }
\refine
\Pi^{=}_{V\setminus V',-V'}
\end{multline*}
Thus, it follows that the \emph{total} time to execute \emph{all} iterations of the \wwhile\ loop \emph{for a specific choice} of $v_i$ 
in Step~3 is of the order of $n$ times the time taken to solve a problem of the following kind:

\begin{quote}
\em
for a subset of nodes $\emptyset \subset V_1\subset V$,  
given 
$\Pi^{=}_{V\setminus V_1,-V_1}$
and a node $v_j\in V\setminus V_1$, compute 
$\Pi^{=}_{V\setminus \left( V_1 \cup \left\{v_j\right\} \right),-\left( V_1 \cup \left\{v_j\right\} \right)}$.
\end{quote}

\noindent
Since 
$\Pi^{=}_{V\setminus \left( V_1 \cup \left\{v_j\right\} \right),-\left( V_1 \cup \left\{v_j\right\} \right)}$
is a refinement of 
$\Pi^{=}_{V\setminus V_1,-V_1}$
by Corollary~\ref{cor1}, we can use the following simple strategy.
For every set 
$S\in\Pi^{=}_{V\setminus V',-V'}$, 
we split $S\setminus\left\{v_j\right\}=\left\{ v_{i_1},v_{i_2},\dots,v_{i_s} \right\}$ 
into two or more parts, if needed, by doing a bucket-sort (with $n$ bins) 
in $O(n\,|S|)$ time 
on the sequence of values 
$
\dist_{v_{i_1},v_j},
\dist_{v_{i_2},v_j},
\dots,
\dist_{v_{i_s},v_j},
$.
The total time taken for all sets in 
$\Pi^{=}_{V\setminus V',-V'}$ is thus 
$
\sum_{S\in\Pi^{=}_{V\setminus V',-V'}} 
O\left(
n \, |S| 
\right)
=O\left( n^2 \right)
$.
\end{proof}

This completes the proof for \mad. Now we consider the claim for \nokmad.
Obviously, \nokmad\ can be solved in $O\left(n^5\right)$ time by solving 
\mad\ for $k=n-1,n-2,\dots,1$ in this order and selecting the largest $k$ as $\kopt$ for which 
$\loptgt<\infty$.
However, we can modify the steps of Algorithm~I directly to solve \nokmad\ in 
$O\left(n^4\right)$ time,
as shown in Algorithm~II.

\medskip
\begin{longtable}{r l l }
\toprule
\multicolumn{3}{c}{Algorithm II: $O\left(n^4\right)$ time deterministic algorithm for \nokmad}
%
\\
\multicolumn{3}{c}
{
(changes from Algorithm-I are shown enclosed in 
\raisebox{0.1cm}{\fbox{$\,\,\,\,\,\,$}}$\,$)
}
\\
\midrule
%
%
{\bf 1.} & 
\multicolumn{2}{l}{
Compute $\vecd_i$ for all $i=1,2,\dots,n$ in $O\left(n^3\right)$ time using Floyd-Warshall algorithm~\cite[p. 629]{CLR90}
}
\\
[5pt]
{\bf 2.} & 
\multicolumn{2}{l}{
$\hatvoptgt\leftarrow\emptyset$ ; $\boxed{\hatkopt\leftarrow 0}$ 
}
\\
[5pt]
{\bf 3.} & 
\multicolumn{2}{l}{
\ffor\ each $v_i\in V$ \ddo 
  \hspace*{0.2in} $(*$ we \emph{guess} $v_i$ to belong to $\voptgt$ $*)$
}
\\
[5pt]
{\bf 3.1} & & 
  $V'=\left\{ v_i \right\}$ 
\\
[5pt]
{\bf 3.2} & & 
\wwhile\ $\boxed{\big(V\setminus V'\neq\emptyset\big)}$ \ddo 
\\
[5pt]
{\bf 3.2.1}
& & 
\hspace*{0.2in}
compute 
$\mu \left(  \cD_{V\setminus V',-V'} \right)$
\\
[5pt]
{\bf 3.2.2}
& & 
\hspace*{0.2in}
\iif\ $\boxed{\big( \, \mu \left(  \cD_{V\setminus V',-V'} \right) > \hatkopt \, \big)}$
\\
{\bf 3.2.3}
& & 
\hspace*{0.4in}
\tthen\ $\,\,\,\,$
$\boxed{\hatkopt\leftarrow \mu \left(  \cD_{V\setminus V',-V'} \right)}$ ; 
$\hatvoptgt\leftarrow V'$
%
\\
[5pt]
{\bf 3.2.4}
& & 
\hspace*{0.4in}
\eelse\
$\,\,\,\,\,\,\,\,$
let $V_1,V_2,\dots,V_\ell$ be the \emph{only} $\ell>0$ equivalence classes (subsets of nodes) 
\\
& & 
\hspace*{1.4in}
in $\Pi_{{V\setminus V',-V'}}^{=}$
such that 
$\left|V_1\right|=\left|V_2\right|=\dots=\left|V_\ell\right|=\mu \left( \cD_{V\setminus V',-V'} \right)$
\\
[5pt]
{\bf 3.2.5}
& & 
\hspace*{0.85in}
$V'\,\leftarrow\,V'\cup \left( \cup_{t=1}^\ell V_t \right)$
\\
[5pt]
{\bf 4.} & 
\multicolumn{2}{l}{
\rreturn\ 
$\boxed{\hatkopt}$ and $\hatvoptgt$ as our solution
}
\\
\bottomrule
\end{longtable}
\medskip

The proof of correctness is very similar (and, in fact simpler due to elimination of some cases) to that of 
\mad.

\bigskip

\noindent
{\bf (b)}
Our solution is the obvious randomization of Algorithm~II (for \mad) or Algorithm-II (for \nokmad) as shown below.

\begin{longtable}{l l l l}
\toprule
\multicolumn{4}{c}{Algorithm III (resp. Algorithm-IV): $O\left(\frac{n^4 \,\log n}{k}\right)$ time 
     randomized algorithm for \mad\ (resp. \nokmad)}
%
\\
\midrule
{\bf 1.} & 
\multicolumn{3}{l}{
Compute $\vecd_i$ for all $i=1,2,\dots,n$ in $O\left(n^3\right)$ time using Floyd-Warshall algorithm
}
\\
[5pt]
{\bf 2.} & 
\multicolumn{3}{l}{$\hatloptgt\leftarrow\infty$ ; $\hatvoptgt\leftarrow\emptyset$ \hspace*{0.15in} (for \mad)}
\\
[3pt]
& or & & 
\\
[3pt]
&
\multicolumn{3}{l}{$\hatvoptgt\leftarrow\emptyset$ ; $\hatkopt\leftarrow 0$ \hspace*{0.2in} (for \nokmad)}
\\
[5pt]
{\bf 3.} & 
\multicolumn{3}{l}{\rrepeat\ $\left\lceil\frac{2 n\ln n}{k}\right\rceil$ times}
\\
[5pt]
& {\bf 3.1} & \multicolumn{2}{l}{select a node $v_i$ uniformly at random from the $n$ nodes} 
\\
[5pt]
& {\bf 3.2} & execute Step 3.1 and Step~3.2 (and its sub-steps) of Algorithm~I & (for \mad)
\\
[3pt]
& & or & 
\\
[3pt]
& &         execute Step 3.1 and Step~3.2 (and its sub-steps) of Algorithm~II & (for \nokmad)
\\
[5pt]
{\bf 4.} & 
\multicolumn{3}{l}{\rreturn\ the best of all solutions found in Step~3}
\\
\bottomrule
\end{longtable}

The success probability $p$ is given by 
\begin{multline*}
p
=
\Pr\left[ \text{$v_i\in\voptgt$ in \emph{at least one} of the $\left\lceil\frac{2 n\ln n}{k}\right\rceil$ iterations} \right]
\\
=
1-
\Pr\left[\!\!\text{$v_i\notin\voptgt$ in \emph{each} of the $\left\lceil\frac{2 n\ln n}{k}\right\rceil$ iterations}\!\!\right]
\geq
1- \left( 1 - \frac{k}{n} \right)^{ \left\lceil\frac{2n\ln n}{k}\right\rceil }
>
1 - \frac{1}{\bee^{2\ln n}}
=
1 - \frac{1}{n^2}
\end{multline*}



\section{Proof of Theorem~\ref{k-large}}

\noindent
{\bf (a)} 
\eqmad\ trivially belongs to $\NP$ for any $k$, thus 
we need to show that it is also $\NP$-hard.

The standard $\NP$-complete \emph{minimum dominating set} (\mds) problem for a graph is defined as follows~\cite{GJ79}.
Our input is a connected undirected unweighted graph $G=(V,E)$. A subset of nodes $V'\subset V$ is called a 
\emph{dominating set} if and only if every node in $V\setminus V'$ is adjacent to some node in $V'$.
The objective of \mds\ is to find a dominating set of nodes of \emph{minimum} cardinality.
Let $\nu(G)$ denote the cardinality of a minimum dominating set for a graph $G$.
It is well-known that the \mds\ and \bSC\ problems have precisely the same approximability via approximation-preserving 
reductions in both directions and, in particular, there exists a standard reduction from \bSC\ to \mds\ as follows.
Given an instance 
$\cU=\left\{a_1,a_2,\dots,a_n\right\}$
and 
$S_1,S_2,\ldots,S_m\subseteq \cU$
of \bSC, we create the following instance $G_1=\left(V_1,E_1\right)$ of \mds. $V_1$ has an \emph{element node}
$v_{a_i}$ for every element $a_i\in\cU$ and a \emph{set node} $v_{S_j}$ for every set $S_j$ with $j\in\{1,2,\dots,m\}$. 
There are two types of edges in $E_1$.
Every set node $v_{S_j}$ has an edge to every other set node $v_{S_\ell}$ and the collection of these edges is
called the set of \emph{clique edges}. Moreover, a set node $v_{S_j}$ is connected to an element node $v_{a_i}$ if and only if
$a_i\in S_j$ and the collection of these edges is called the set of \emph{membership edges}. A standard straightforward 
argument shows that $\cI\subset \{1,2,\dots,m\}$ is a solution of \bSC\ if and only if 
the collection of set nodes $\left\{ \,v_{S_i} \,|\, i\in\cI \, \right\}$ is a solution of \mds\ on $G_1$ 
and thus $\scopt=\nu\left(G_1\right)$.

For the purpose of our $\NP$-hardness reduction, it would be more convenient to work with a restricted version 
of \bSC\ known as the \emph{exact cover by $3$-sets} (\xtc) problem. 
Here we have exactly $n$ elements and exactly $n$ sets where $n$ is a multiple of three, every set contains exactly $3$ elements and 
every element occurs in exactly $3$ sets. Obviously we need at least $\frac{n}{3}$ sets to cover all the $n$ elements. 
Letting $\xtcopt$ to denote the number of sets in an optimal solution of \xtc,
it is well-known that problem of deciding whether $\xtcopt=\frac{n}{3}$ is in fact
$\NP$-complete.

Let $n_1=\frac{-6k+\sqrt{36k^2+24(n-k)}}{4}$ be the real-valued solution of the quadratic equation 
$n_1\left( 2k+\frac{2n_1}{3} \right) + k =n$.
Note that since $k\leq n^\eps$ for some constant $\eps<\frac{1}{2}$, 
we have 
$n_1=\Theta\left(\sqrt{n}\,\right)
$, 
\IE, $n$ and $n_1$ are ``polynomially related''.

%
\begin{figure}[ht]
\centering
\includegraphics[width=\linewidth]{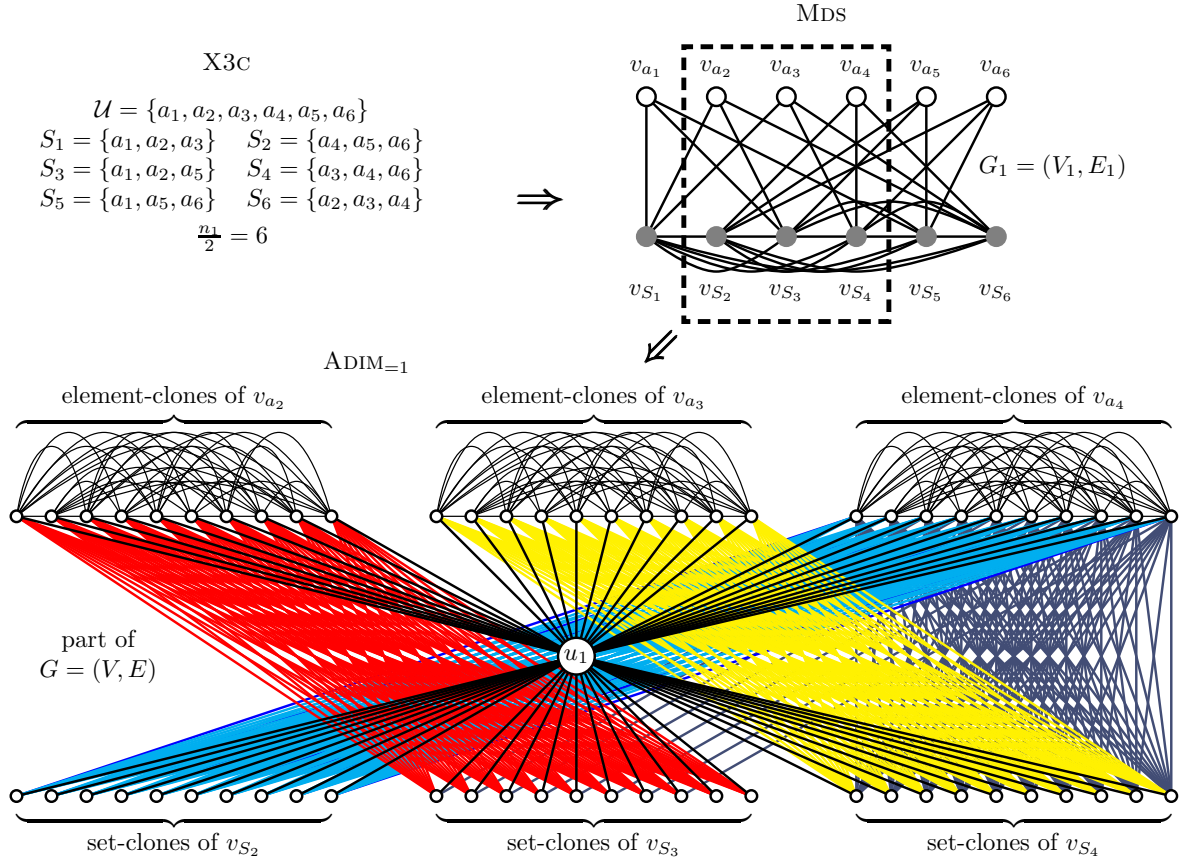}
\caption{\label{fig2}Illustration of the $\NP$-hardness reduction in Theorem~\ref{k-large}(a). Only a part of the 
graph $G$ is shown for visual clarity.}
\end{figure}

We assume without loss of generality that $n_1$ is an even integer, 
and start with an instance of \xtc\ of $\frac{n_1}{2}$ elements 
and transform it to an instance graph $G_1=\left(V_1,E_1\right)$ having $n_1$ nodes of \mds\ via the reduction outlined before.
Since $\frac{n_1}{2}$ is polynomially related to $n$,
such an instance of \xtc\ is $\NP$-complete with respect to $n$ being the input size.
We reduce $G_1$ to an instance $G=(V,E)$ of \eqmad\ in polynomial time as follows
(see \FI{fig2} for an illustration):
\begin{itemize}
\item
We ``clone'' each element node $v_{a_j}\in V_1$ to get 
$2k+\frac{2n_1}{3}$ copies, \IE, 
every node 
$v_{a_j}$
is replaced by 
$2k+\frac{2n_1}{3}$
new nodes
$
v_{a_j,1},
v_{a_j,2},\dots,
v_{a_j,2k+\frac{2n_1}{3}},
$.
We refer to these nodes as \emph{clones} of the element node 
$v_{a_j}$ 
(or, sometimes simply as \emph{element-clone nodes}).
There are precisely
$
n_1 \left( k+\frac{n_1}{3} \right)
$
such nodes.
\item
We ``clone'' each set node $v_{S_j}\in V_1$ to get 
$2k+\frac{2n_1}{3}$ copies, \IE, 
every node 
$v_{S_j}$
is replaced by 
$2k+\frac{2n_1}{3}$
new nodes
$
v_{S_j,1},
v_{S_j,2},\dots,
v_{S_j,2k+\frac{2n_1}{3}},
$.
We refer to these nodes as \emph{clones} of the set node 
$v_{S_j}$
(or, sometimes simply as \emph{set-clone nodes}).
There are precisely
$
n_1 \left( k+\frac{n_1}{3} \right)
$
such nodes.
\item
We add $k$ new nodes 
$u_1,u_2,\dots,u_k$. 
We refer to these nodes as \emph{clique nodes}.
\item
We add an edge between every pair of clique nodes $u_i$ and $u_j$. 
We refer to these edges as \emph{clique edges}. There are precisely $\binom{k}{2}$ such edges.
\item
We add an edge between every clique node and every non-clique node, \IE, we add every edge in the set 
\begin{multline*}
\Big\{ \left\{ u_i,v_{a_j,\ell} \right\} \,|\, 1\leq i\leq k,\, 1\leq j\leq \frac{n_1}{2},\,1\leq\ell\leq 2k+\frac{2n_1}{3} \Big\}
\\
\,\bigcup\,
\Big\{ \left\{ u_i,v_{S_j,\ell} \right\} \,|\, 1\leq i\leq k,\, 1\leq j\leq \frac{n_1}{2},1\leq\ell\leq 2k+\frac{2n_1}{3} \Big\}
\end{multline*}
We refer to these edges as the \emph{partition-fixing} edges. 
There are precisely 
$
kn_1 \left( k+\frac{n_1}{3}\right)
$
such edges.
\item
We add an edge between every pair of distinct element-clone nodes 
$v_{a_j,\ell}$ and 
$v_{a_{j'},{\ell'}}$. 
We refer to these as the \emph{element-clone edges}. 
There are precisely 
$
\binom{ 2k+ ({2n_1})/{3} }{2}
$
such edges.
\item
For every element $a_i$ and every set $S_j$ such that $a_i\notin S_j$, we add 
the following $\left(2k+\frac{2n_1}{3}\right)^2$ edges:
\[
\left\{ v_{S_j,\ell}, v_{a_i,p}\right\} 
\,\,\,\text{for}\,\,\,
1\leq \ell,p \leq 
2k+\frac{2n_1}{3}
\]
We refer to these edges as the \emph{non-member} edges
corresponding to the element node $a_i$ and the set node $S_j$.
There are precisely 
$
\frac{3n_1}{2}\left(2k+\frac{2n_1}{3}\right)^2
$
such edges.
\end{itemize}
Note that $G$ has precisely 
$n_1\left(2k+\frac{2n_1}{3}\right)+k=n$ 
nodes and thus 
our reduction is polynomial time in $n$.
Since any clique node is adjacent to every other node in $G$, it follows that 
$\diam(G)=2$.
We now show the validity of our reduction by showing that 
\[
(\star)\,\,
\nu\left(G_1\right)=\frac{n_1}{3} \,\,\text{if and only if}\,\,
\lopteq\leq\frac{n_1}{3}
\]

\smallskip

\noindent
{\bf Proof of $\nu\left(G_1\right)=\frac{n_1}{3}\,\Rightarrow\, \lopteq\leq\frac{n_1}{3}$}

\medskip

Consider an optimal solution $V_1'\subset \left\{ v_{S_1},v_{S_2},\dots,v_{S_{n_1}}\right\}$ of \mds\ on $G_1$
with 
$\nu\left(G_1\right)=\left|V_1'\right|=\frac{n_1}{3}$.
We now construct a solution $V'\subset V$ of \eqmad\ on $G$ by setting 
$V'=\left\{ v_{S_j,1} \,|\, v_{S_j} \in V_1'\right\}$. 
Note that $|V'|=\left|V_1'\right|=\frac{n_1}{3}$.
We claim that $V'$ is a valid solution of \eqmad\ by showing that 
\begin{description}
\item[(\emph{a})]
$\left\{ u_1,u_2,\dots,u_k\right\}\in\Pi^{=}_{V\setminus V',-V'}$ and 
\item[(\emph{b})]
any other equivalence class in 
$\Pi^{=}_{V\setminus V',-V'}$ has at least $k$ nodes.
\end{description}
To prove (\emph{a}),
consider a clique node $u_i$ and any other non-clique node. Then, the following cases apply:
\begin{itemize}
\item
Suppose that the non-clique node is a element-clone node 
$v_{a_j,\ell}\in V\setminus V'$ 
for some $j$ and $\ell$.
Since $V_1'$ is a solution of \mds\ on $G_1$, there exists a set node $v_{S_p}\in V_1'$ such that 
$\left\{ v_{S_p},v_{a_j}\right\}\in E_1$
and consequently 
$\left\{ v_{S_p,1},v_{a_j,\ell}\right\}\notin E$.
This implies 
that there exists a node 
$v_{S_p,1}\in V'$
such that 
$1=\dist_{u_i,v_{a_j,\ell}}\neq
\dist_{v_{S_p,1},v_{a_j,\ell}}$, 
and therefore 
$v_{a_j,\ell}$ \emph{cannot} be in the same equivalence class with $u_i$.
\item
Suppose that the non-clique node is a 
set-clone node 
$v_{S_j,p}\in V\setminus V'$.
Pick any 
set-clone node $v_{S_\ell,1}\in V'$.
Then, 
$1=\dist_{u_i,v_{S_j,p}}\neq \dist_{v_{S_j,p},v_{S_\ell,1}}$, 
and therefore 
$v_{S_j,p}$ \emph{cannot} be in the same equivalence class with $u_i$.
\end{itemize}
To prove (\emph{b}),
note the following:
\begin{itemize}
\item
Since $\diam(G)=2$,
$\dist_{v_{S_i,p},v_{S_j,q}}=2$ for any two \emph{distinct} set-clone nodes $v_{S_i,p}$ and $v_{S_j,q}$,
and thus all the set nodes in $V\setminus V'$ belong together in the \emph{same} equivalence class in 
$\Pi^{=}_{V\setminus V',-V'}$.
There are at least $n_1 \left( k + \frac{n_1}{3} \right) - \frac{n_1}{3}>k$ such nodes in $V\setminus V'$.
Thus, any equivalence class that contains these set-clone nodes cannot have 
less than $k$ nodes.
\item
Consider now an equivalence class in 
$\Pi^{=}_{V\setminus V',-V'}$
that contains 
a copy 
$v_{a_i,j}$ of the element node $v_{a_i}$ for some $i$ and $j$.
Consider 
another copy 
$v_{a_i,\ell}$ of the element node $v_{a_i}$ for some $\ell\neq j$.
For any set node 
$v_{S_p,1}\in V'$, 
if $a_i\notin S_p$ then 
$
\dist_{v_{S_p,1},v_{a_i,j}}
=
\dist_{v_{S_p,1},v_{a_i,\ell}}
=1
$, 
whereas 
if $a_i\in S_p$ then, 
since $\diam(G)=2$,
it follows that
$
\dist_{v_{S_p,1},v_{a_i,j}}
=
\dist_{v_{S_p,1},v_{a_i,\ell}}
=2
$. 
Thus, any equivalence class that contains at least one clone of an element node must contain 
all the $2k+\frac{2n_1}{3}>k$ clones of that element node and thus 
such an equivalence class cannot have a number of nodes that is less than $k$.
\end{itemize}

\bigskip

\noindent
{\bf Proof of 
$
\lopteq\leq\frac{n_1}{3}
\,\Rightarrow\,
\nu\left(G_1\right)=\frac{n_1}{3}
$
}

\medskip

Since we know that 
$\nu\left(G_1\right)$ is always at least $\frac{n_1}{3}$, 
it suffices to show that
$
\lopteq\leq\frac{n_1}{3}
\,\Rightarrow\,
\nu\left(G_1\right)\leq\frac{n_1}{3}
$.
Consider an optimal solution $\vopteq\subset V$ with $\lopteq\leq\left|\vopteq\right|=\frac{n_1}{3}$. 
Since $\vopteq$ is a solution of \eqmad\ on $G$, there exists a subset of nodes, say 
$\vhat\subset V\setminus\vopteq$, such that $|\vhat|=k$ and 
$\vhat\in 
\Pi^{=}_{V\setminus \vopteq,-\vopteq}
$.

\begin{proposition}\label{no-element-set}
$\vhat$ does not contain any set-clone or element-clone nodes and thus 
$\vhat=\left\{u_1,u_2,\dots,u_k\right\}$.
\end{proposition}

\begin{proof}
Suppose that $\vhat$ contains at least one element-clone node $v_{a_i,j}$ for some $i$ and $j$.
But, $V\setminus\vopteq$ contains at least 
$2k+\frac{2n_1}{3}-\frac{n_1}{3}-1>k$ other clones of the element node $a_i$ and all these clones
must belong together with $v_{a_i,j}$ in the \emph{same} equivalence class.
This implies 
$|\vhat|\geq 2k+\frac{2n_1}{3}-\frac{n_1}{3}>k$, a contradiction.

Similarly, suppose that $\vhat$ contains at least one set-clone node $v_{S_i,j}$ for some $i$ and $j$.
But, $V\setminus\vopteq$ contains at least 
$2k+\frac{2n_1}{3}-\frac{n_1}{3}-1>k$ other clones of the set node $S_i$ and all these clones
must belong together with $v_{S_i,j}$ in the \emph{same} equivalence class.
This implies 
$|\vhat|\geq 2k+\frac{2n_1}{3}-\frac{n_1}{3}>k$, a contradiction.
\end{proof}

\begin{proposition}\label{no-set-clone}
$\vopteq$ does not contain two or more clones of the same set node.
\end{proposition}

\begin{proof}
Suppose that $\vopteq$ contains two set-clone nodes $v_{S_j,p}$ and $v_{S_j,q}$ of the same set node $v_{S_j}$.
But, $V\setminus\vopteq$ contains at least 
$2k+\frac{2n_1}{3}-\frac{n_1}{3}-1>k$ other clones of the element node $a_i$ and all these clones
must belong together in the \emph{same} equivalence class $S$.
If we remove 
$v_{S_j,p}$ from $\vopteq$ then $v_{S_j,p}$ gets added to this equivalence class. 
Thus, such a removal produced another valid solution but with one node less than $\lopt$, contradicting 
the optimality of $\lopteq$.
\end{proof}


\begin{proposition}\label{no-element-clone}
$\vopteq$ does not contain any element-clone node.
\end{proposition}

\begin{proof}
Suppose that $\vopteq$ contains at least one element-clone node and thus at most $\frac{n_1}{3}-1$ 
set-clone nodes. 
Note that $V\setminus\vopteq$ contains at least 
$2k+\frac{2n_1}{3}-\frac{n_1}{3}$ clones of every element node $a_i$.
Consider an element-clone node 
$v_{a_i,p}\in V\setminus\vopteq$ and a clique node $u_j$.
Since 
$\vhat=\left\{u_1,u_2,\dots,u_k\right\}\in 
\Pi^{=}_{V\setminus \vopteq,-\vopteq}
$, 
there must be a node in $\vopteq$ such that the distance of this node to $u_j$ is different from the distance to
$v_{a_i,p}$. Such a node in $\vopteq$ cannot be an element-clone node, say  
$v_{a_\ell,q}$ since 
$
\dist_{
v_{a_i,p},
v_{a_\ell,q}
}
=
\dist_{
u_j,
v_{a_\ell,q}
}
=1
$.
Since there is an edge between every set-clone node and every clique node, 
such a node must be a set-clone node, say $v_{S_r,s}$ for some $r$ and $s$, such that 
$
\dist_{
v_{a_i,p},
v_{S_r,s}
}
=2
$, \IE, $a_i\in S_r$.
Since every set in \xtc\ contains exactly $3$ elements and $3\times \left(\frac{n_1}{3}-1\right)<n_1$, 
there must then exist an element-clone node $v_{a_i,p}$ such that the distance of $v_{a_i,p}$ to any node in $\vopteq$ is exactly 
the same as the distance of $u_j$ to that node in $\vopteq$. This implies 
$v_{a_i,p}\in\vhat$, contradicting
Proposition~\ref{no-element-set}.
\end{proof}


By Proposition~\ref{no-set-clone}
and 
Proposition~\ref{no-element-clone}, $\vopteq$ contains exactly one clone of a subset of set nodes.
Without loss of generality, assume that 
$
\vopteq = 
\left\{
v_{S_j,1} \,|\, j\in J,\, J\subset \left\{1,2,\dots,\frac{n_1}{2} \right\} \,
\right\}
$
and let 
$
V_1'=\left\{
v_{S_j} \,|\, v_{S_j,1}\in\vopteq \,
\right\}
$.
Note that $\left|V_1'\right|=\left|\vopteq\right|$.
We are now ready to finish our proof by showing $V_1'$ is indeed a valid solution of \mds\ on $G_1$.
Suppose not, and let $v_{a_i}$ be an element-node that is not adjacent to any node in $V_1'$. Then, 
\begin{multline*}
\forall\,v_{S_j}\in V_1'\,:\,\left\{ v_{a_i},v_{S_j} \right\}\notin E_1 
\,\Rightarrow\,
\forall\,v_{S_j,1}\in \vopteq\,:\,\left\{ v_{a_i,1},v_{S_j,1} \right\}\in E 
\\
\,\Rightarrow\,
\forall\,v_{S_j,1}\in \vopteq\,:\,\dist_{v_{a_i,1},v_{S_j,1}}=1
\,\Rightarrow\,
v_{a_i,1}\in\vhat
\end{multline*}
which contradicts Proposition~\ref{no-element-set}.

\medskip
\noindent
{\bf (b)}
The proof is similar to that of (a) 
but this time we start with a general version of \bSC\ as opposed to the restricted 
\xtc\ version, 
and show that the reduction is approximation-preserving in an appropriate sense.
In the sequel, we use the standard notation poly$(n)$ to denote a polynomial $n^c$ of $n$ (for some constant $c>0$).
We recall the following details of the inapproximability reduction of Feige in~\cite{F98}.
Given an instance formula $\phi$ of the standard Boolean satisfiability problem (\sat), 
Feige reduces $\phi$ to an instance $\cU,S_1,S_2,\dots,S_m$ of \bSC\ 
(with $m=\mbox{poly}(n)$)
in $O(n^{\log\log n})$ time such that the following properties are satisfied for any constant $0<\eps<1$:
\begin{itemize}
\item
For some $Q>0$, either 
$\scopt=\frac{n}{Q}$
or
$\scopt> 
\left(\frac{n}{Q}\right)
(1-\eps)
\ln n$.
\item
The reduction satisfies the following completeness and soundness properties:
\begin{description}
\item[\hspace*{0.5in}(completeness)]
If $\phi$ is satisfiable then 
$\scopt=\frac{n}{Q}$.
\item[\hspace*{0.7in}(soundness)]
If $\phi$ is not satisfiable then 
$\scopt>
\left(\frac{n}{Q}\right)
(1-\eps)
\ln n$.
\end{description}
\end{itemize}
%
Since $m=\mbox{poly}(n)$,
by adding duplicate copies of a set, if necessary, we can ensure that $m=n^c-n$ for some constant $c\geq 1$.
Our reduction from \bSC\ to \mds\ to \eqmad\ is same as in~(a) except that some details are different, which we 
show here.
\begin{itemize}[leftmargin=*]
\item
We start with an instance of \bSC\ as given by Feige in~\cite{F98} with $n_1$ elements and $m=(n_1)^c-n_1$ sets, where 
$n_1=
\left( \frac{ -k + \sqrt{k^2 + 2(n-k) }  }   {2} \,\right)^{1/c}
$ 
is a real-valued solution of the equation 
$
(n_1)^{\,2c} + k (n_1)^c - \frac{n-k}{2}=0
$.
Note that since $k\leq n^\eps$ for some constant $\eps<\frac{1}{2}$, 
we have 
$n_1=\Theta\left( {n}^{1/(2\,c)}\right)
$, 
\IE, $n$ and $n_1$ are polynomially related.
\item
We make $2(n_1)^c+2k$ copies of each element node and each set node as opposed to 
$2k+\frac{2n_1}{3}$ copies that we made in 
the proof of (a).
Note that $G$ has again precisely 
$(n_1)^c\left(2k+2(n_1)^c\,\right)+k=n$ 
nodes.
\item
Let $\delta>0$ be the constant given by $\delta=\frac{\ln n}{(1-\eps)\ln n_1}$.
Our claim $(\star)$ in 
the proof of (a)
is now modified to 

\smallskip
%
%
\begin{tabular}{l} 
$(\star)$  
\begin{tabular}{r l} 
(completeness)
&
if
$\nu\left(G_1\right)=\frac{n_1}{Q}$ 
then 
$\lopteq\leq\frac{n_1}{Q}$
\\
[4pt]
(soundness)
&
if 
$\nu\left(G_1\right)
>
\left(\frac{n_1}{Q}\right)
(1-\eps)
\ln n_1$
then 
$\lopteq> 
\left(\frac{n_1}{Q}\right)
(1-\eps)
\ln n_1
=
\left(\frac{n_1}{Q}\right)
\frac{1}{\delta}
\ln n
$
\end{tabular}
\end{tabular}
%
\item
Our proof of the \emph{completeness} claim follows the 
``{Proof of $\nu\left(G_1\right)=\frac{n_1}{3}\,\Rightarrow\, \lopteq\leq\frac{n_1}{3}$}'' in the proof of
(a) 
with the obvious 
replacement of $\frac{n_1}{3}$ by $\frac{n_1}{Q}$.
\item
Note that our soundness claim is equivalent to its contra-positive 
\[
\text{if}\,\,\, 
\lopteq \leq
\left(\frac{n_1}{Q}\right)
(1-\eps)
\ln n_1
\,\,\,\text{then}\,\,\,
\nu\left(G_1\right)
\leq
\left(\frac{n_1}{Q}\right)
(1-\eps)
\ln n_1
\]
and the proof of this contra-positive follows the 
``Proof of 
$
\lopteq\leq\frac{n_1}{3}
\,\Rightarrow\,
\nu\left(G_1\right)=\frac{n_1}{3}
$'' in 
the proof of (a).
In the proof, 
the quantity 
$2k+\frac{2n_1}{3}$ corresponding to the number of copies for each set and element node needs to be replaced by  
$2(n_1)^c+2k$; note that 
$(2(n_1)^c+2k)-n_1\gg k$.
\end{itemize}

\medskip
\noindent
{\bf (c)}
Since $k=n-c$ for some constant $c$, 
$\Pi^{=}_{V\setminus \vopteq,-\vopteq}$
contains a single equivalence class 
$V'\subset V$ such that $|V'|=k$. 
Thus, we can employ the straightforward exhaustive method of 
selecting every possible subset $V'$ of $k$ nodes to be in 
$\Pi^{=}_{V\setminus V',-V'}$
and checking if the chosen subset of nodes provide a valid solution. 
There are $\binom{n}{k}<n^c$ such possible subsets and therefore 
the asymptotic running time is $O\left(n^c+n^3\right)$ which is polynomial in $n$.
Note that for this case $\lopteq=c$ if a solution exists.



\section{Proof of Theorem~\ref{k-one}}
\label{sec-pfe}

\noindent
{\bf (a)}
Note that trivially $\lopteqone\leq n-1$ and thus $\vopteqone\neq\emptyset$.
Our algorithm, shown as Algorithm~V, uses the greedy logarithmic approximation of Johnson~\cite{J74} 
for \bSC\ that selects, at each successive step, a set that contains the maximum number of 
elements that are still not covered.

\smallskip
\begin{longtable}{l l l }
\toprule
\multicolumn{3}{c}{Algorithm V: $O\left(n^3\right)$-time $(1+\ln(n-1)\,)$-approximation algorithm for \eqmadone.}
%
\\
\midrule
%
%
{\bf 1.} & 
\multicolumn{2}{l}{
Compute $\vecd_i$ for all $i=1,2,\dots,n$ in $O\left(n^3\right)$ time using Floyd-Warshall algorithm.
}
\\
[5pt]
{\bf 2.} & 
\multicolumn{2}{l}{
$\hatlopteqone\leftarrow\infty$ ; 
$\hatvopteqone\leftarrow\emptyset$
}
\\
[5pt]
{\bf 3.} & 
\multicolumn{2}{l}{
\ffor\ each node $v_i\in V$ \ddo 
  \hspace*{0.4in} $(*$ we \emph{guess} the set $\left\{ v_i \right\}$ to belong to  $\Pi^{=}_{V\setminus \vopteqone,-\vopteqone}$ $*)$
}
\\
[5pt]
& {\bf 3.1} & 
      create the following instance of \bSC\ containing $n-1$ elements and $n-1$ sets:
\\
& & 
	 \hspace*{0.2in}
   $\cU=\left\{ \,a_{v_j} \,|\, v_j\in V\setminus \left\{v_i\right\} \, \right\}$, 
\\
& & 
	 \hspace*{0.2in}
	 $S_{v_j}= \left\{ a_{v_j} \right\} \cup \left\{ \,a_{v_\ell} \,|\, \dist_{v_i,v_j}\neq\dist_{v_\ell,v_j}\right\}$ 
	        for $j\in \{1,2,\dots,n\}\setminus \{i\}$
\\
[5pt]
& {\bf 3.2} & 
\iif\ $\cup_{j\in \{1,2,\dots,n\}\setminus \{i\}}S_{v_j}=\cU$ \tthen
\\
[5pt]
& & {\bf 3.2.1}  
\hspace*{0.05in}
run the greedy approximation algorithm~\cite{J74} for this instance of \bSC\
\\
& & 
	 \hspace*{0.8in}
	  giving a 
		solution $\cI\subseteq \{1,2,\dots,n\}\setminus\{i\}$ 
\\
[5pt]
& & {\bf 3.2.2}
\hspace*{0.05in}
        $V'=\left\{ \, v_j \,|\, j\in\cI \, \right\}$
\\
[5pt]
& & {\bf 3.2.3}
\hspace*{0.05in}
   \iif\ 
   $\big( \, |V'|<\hatlopteqone \, \big)$ 
   \tthen\
   $\,\,\,\hatlopteqone\leftarrow \left|V'\right|$ ; 
   $\hatvopteqone\leftarrow V'$
\\
[5pt]
{\bf 4.} & 
\multicolumn{2}{l}{
\rreturn\ $\hatlopteqone$ and $\hatvopteqone$ as our solution
}
\\
\bottomrule
\end{longtable}

\begin{lemma}[Proof of correctness]
Algorithm~V returns a valid solution for \eqmadone.
\end{lemma}

\begin{proof}
%
Suppose that our algorithm returns an invalid solution in 
the iteration of the \ffor\ loop in Step~3 when $v_i$ is equal to $v_\ell$ for some $v_\ell\in V$.
We claim that this cannot be the case since 
$\left\{ v_\ell \right\}\in\Pi^{=}_{V\setminus V',-V'}$.
Indeed, since $\cI$ is a valid solution of the \bSC\ instance, 
for every $j\notin \left\{ \ell \right\} \cup \cI$, the following holds:
\begin{gather*}
\exists \, t\in\cI \,:\, a_{v_j}\in S_{v_t} 
\,\Rightarrow\,
\exists \, v_t\in V' \,:\, 
	 \dist_{v_\ell,v_t}\neq\dist_{v_j,v_t}
\end{gather*}
and thus $v_\ell$ cannot be together with any other node in any equivalence class in 
$\Pi^{=}_{V\setminus V',-V'}$.
\end{proof}

\begin{lemma}[Proof of approximation bound]
Algorithm~V 
solves \eqmadone\ with an approximation ratio of $1+\ln (n-1)$.
\end{lemma}

\begin{proof}
Fix any optimal solution $\vopteqone$.
Since $\mu \left(  \cD_{V\setminus \vopteqone,-\vopteqone} \right)=1$,  
$\left\{ v_\ell \right\}\in\Pi^{=}_{V\setminus \vopteqone,-\vopteqone}$
for some $v_\ell\in V$.
Consider the iteration of the \ffor\ loop in Step~3 when $v_i$ is equal to $v_\ell$.
We now analyze the run of \emph{this particular iteration}, and claim that the set-cover instance 
created during this iteration satisfies 
$\scopt\leq \left|\vopteqone\right|=\lopteqone$.
To see this, construct the following solution of the set-cover instance from $\vopt$ containing 
exactly $\lopt$ sets:
\[
v_i\in\vopteqone \,\,\equiv\,\,i\in\cI 
\]
To see that this is indeed a valid solution of the set-cover instance, consider any 
$a_{v_j}\in \cU=\{a_{v_1},a_{v_2},\dots,a_{v_n}\}\setminus \{a_{v_\ell}\}$.
Then, the following cases apply showing that $a_{v_j}$ belongs to some set selected in our
solution of \bSC:
\begin{itemize}
\item
if $j\in\cI$ then 
$a_{v_j}\in S_{v_j}$ and $S_{v_j}$ is a selected set in the solution.
\item
if $j\notin\cI$ then 
$
v_j\in V\setminus\vopt
\,\Rightarrow\,
\exists\,v_t\in\vopt \,:\,
    \dist_{v_\ell,v_t}\neq\dist_{v_j,v_t}
\,\Rightarrow\,
\exists\,t\in\cI \,:\,
a_{v_j}\in S_{v_t}
$.
\end{itemize}
Using the approximation bound of the algorithm of~\cite{J74} it now follows that the quality of our solution
$\hatlopteqone$ 
satisfies 
\[
\hatlopteqone
=
\left| \hatvopteqone \right| 
=|\cI|
<
(1+\ln (n-1)\,)\scopt
\leq
(1+\ln (n-1)\,)\lopteqone
\]
\end{proof}

\begin{lemma}[Proof of time complexity]
Algorithm~V runs in 
$O\left(n^3\right)$ time.
\end{lemma}

\begin{proof}
There are a total of $n$ instances of set cover that we need to build in Step~3.1 and solve by the greedy heuristic
in Step~3.2.1.
Building the set-cover instance can be trivially done in $O\left(n^2\right)$ time by comparing $\dist_{v_i,v_j}$ for 
all appropriate pairs of nodes $v_i$ and $v_j$.
Since the set-cover instance in Step~3.1 has $n-1$ sets each having no more than $n-1$ elements, each implementation 
of the greedy heuristic in Step~3.2.1 takes $O\left(n^2\right)$ time.
\end{proof}

\smallskip

\noindent
{\bf (b)}
Let $v_i$ be the node of degree $1$.
Let $v_\ell$ be the unique node adjacent to $v_i$ (\IE, $\left\{v_i,v_\ell\right\}\in E$). 
Consider the following solution of \eqonemad: $V'=\left\{v_i\right\}$.
We claim that is a valid solution of \eqonemad\ by showing that 
$\left\{v_{\ell}\right\}\in\Pi^{=}_{V\setminus V',-V'}$.
Consider any node $v_j\in V\setminus \left\{v_i,v_\ell \right\}$,
Then, $1=\dist_{v_\ell,v_i}\neq\dist_{v_j,v_i}$.

%
\begin{figure}[ht]
\centering
\includegraphics[width=\linewidth]{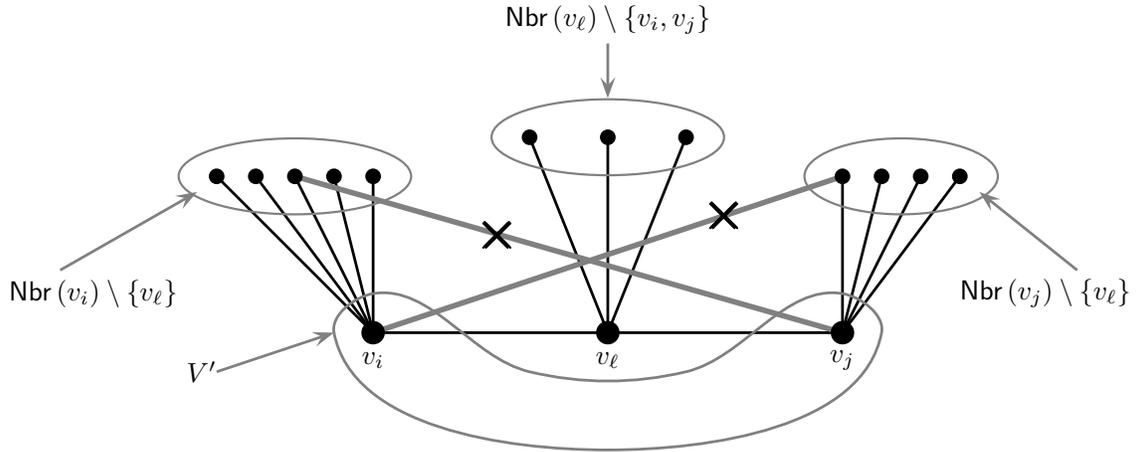}
%
%
\caption{\label{fig1}Illustration of the proof of Theorem~\ref{k-one}(c). Edges marked by 
$\displaystyle\pmb{\times}$ cannot exist. No node in 
$\nbr\left(v_\ell\right)\setminus \left\{v_i,v_j\right\}$
can have an edge to {\em both} $v_i$ and $v_j$.}
\end{figure}

\bigskip

\noindent
{\bf (c)}
Since $G$ does not contain a $4$-cycle, $\diam(G)\geq 2$. 
Thus, 
there exists two nodes $v_i,v_j\in V$ such that $\dist_{v_i,v_j}=2$. 
Let 
$
v_{\ell}
$
be a node at a distance of $1$ from both $v_i$ and $v_j$
on a shortest path between $v_i$ and $v_j$ (see~\FI{fig1}).
Consider the following solution of \eqonemad: $V'=\left\{v_{i},v_{j} \right\}$. 
Note that $v_{\ell}\in V\setminus V'$.
We claim that this is a valid solution of \eqonemad\ by showing that 
$\left\{v_{\ell}\right\}\in\Pi^{=}_{V\setminus V',-V'}$ (\IE,  
no node $v_p \in V \setminus \left\{v_i,v_j,v_\ell\right\}$ can belong together with $v_\ell$ 
in the same equivalence class of 
$\Pi^{=}_{V\setminus V',-V'}$)
in the following manner:
\begin{itemize}
\item
If 
$v_p\in\nbr\left(v_i\right)\setminus \left\{v_\ell\right\}$
then
$\dist_{v_\ell,v_j}=1$
but 
$\dist_{v_p,v_j}\neq 1$
since $G$ has no $4$-cycle (see the edges marked $\displaystyle\pmb{\times}$ in \FI{fig1}).
\item
If 
$v_p\in\nbr\left(v_j\right)\setminus \left\{v_\ell\right\}$
then
$\dist_{v_\ell,v_i}=1$
but 
$\dist_{v_p,v_i}\neq 1$
since $G$ has no $4$-cycle (see the edges marked $\displaystyle\pmb{\times}$ in \FI{fig1}).
\item
If 
$v_p\in\nbr\left(v_\ell\right)\setminus \left\{v_i,v_j\right\}$
then 
$v_p$ cannot be adjacent to \emph{both} $v_i$ and $v_j$ since $G$ does not contain a $4$-cycle.
This implies that 
$\dist_{v_\ell,v_i}=\dist_{v_\ell,v_j}=1$
but at least one of 
$\dist_{v_p,v_i}$ and 
$\dist_{v_p,v_j}$ is not equal to $1$.
\item
If $v_p$ is any node not covered by the above cases, then 
$\dist_{v_p,v_i}>1$
but 
$\dist_{v_\ell,v_i}=1$.
\end{itemize}


\section{Concluding Remarks}
\label{sec-conc}

Prior to our work, 
known results for these privacy measures 
only included some heuristic algorithms with no provable guarantee on performances such as in~\cite{Trujillo-1},
or algorithms for very special cases. In fact, it was not even known if any version of these related computational problems is 
$\NP$-hard. Our work 
provides the first non-trivial computational complexity results for effective computation of these measures.
Theorem~\ref{thm1} shows that both \nokmad\ and \mad\ are \emph{provably} computationally easier problems than 
\eqmad.
In contrast, 
Theorem~\ref{k-large}(a)--(b) and Theorem~\ref{k-one}
show that 
\eqmad\ is in general computationally hard but admits approximations or exact solution for specific choices of $k$ or 
graph topology.
We believe that our results will stimulate further research on quantifying and computing privacy measures for networks.
In particular, our results raise the following interesting research questions:
\begin{enumerate}[label=$\blacktriangleright$]
\item
We have only provided a logarithmic approximation algorithm for \eqmadone. Is it possible to design a non-trivial 
approximation algorithm for \eqmad\ for $k>1$ ? We conjecture that a $O(\log n)$-approximation is possible for \eqmad\
for every fixed $k$.
\item
We have provided a logarithmic inapproximability result for \eqmad\ for every $k$ \emph{roughly} up to $\sqrt{n}$.
Can this approximability result be further improved when $k$ is not a constant ?  
We conjecture that the inapproximability factor can be further improved to $\Omega\left(n^\eps\right)$ for some 
constant $0<\eps<1$ when $k$ is around $\sqrt{n}$.
\end{enumerate}

%


\newpage

\end{document}